\documentclass[runningheads]{llncs}
\usepackage{amsmath,amsfonts} 
\usepackage[svgnames]{xcolor} 

\usepackage{url}
\usepackage{hyperref}
\usepackage{enumerate}
\usepackage{enumitem} 
\usepackage{tikz}
\usepackage{amssymb}
\usetikzlibrary{shapes}

\usepackage{graphicx}
\usepackage{url}
\usepackage{hyperref}
\usepackage{enumerate}
\usepackage{tikz}
\usetikzlibrary{decorations.markings, calc}

\DeclareGraphicsExtensions{.png,.pdf,.jpg}

\usepackage{wrapfig}

\usepackage{etoolbox}

\usepackage[noend]{algpseudocode}
\AtBeginEnvironment{algorithmic}{\scriptsize}

\newcommand{\idx}{I}
\DeclareMathOperator*{\profile}{profile}

\newcommand{\len}[1]{\lvert{#1}\rvert}
\newcommand{\norm}[1]{\lVert{#1}\rVert}

\newcommand{\nil}{\textsc{nil}}

\newcommand{\BDD}{\textsc{bdd}}
\newcommand{\DAG}{\textsc{dag}}
\newcommand{\BDDs}{\textsc{bdd}s}

\newcommand{\True}{{\normalfont\texttt{1}}}
\newcommand{\False}{{\normalfont\texttt{0}}}
\newcommand{\Tnode}{\top}
\newcommand{\Fnode}{\bot}
\def\pp{\mathinner{\ldotp\ldotp}} 
\def\preorder{\prec_{\text{\normalfont pre}}}
\def\postorder{\prec_{\text{\normalfont post}}}

\newcommand{\Qed}{\ensuremath{\hfill \square}}

\usepackage{algorithm}

\newcommand{\funcs}[2]{\textsc{#1}({#2})}

\title{Binary Decision Diagrams:\\ from Tree Compaction to Sampling\thanks{This work
 was partially supported by the \textsc{anr} projects \textsc{Metaconc} ANR-15-CE40-0014 and \textsc{Ping/Ack} ANR-18-CE40-0011.
 An implementation of the results is provided at \url{https://github.com/agenitrini/BDDgen}.}
	}

\author{Julien Cl\'ement\inst{1}
    \and
    Antoine Genitrini\inst{2}}

\institute{
   Normandie Univ, \textsc{unicaen}, \textsc{ensicaen}, \textsc{cnrs}, \textsc{greyc}, 14000 Caen, France
    \email{Julien.Clement@unicaen.fr}
    \and
    Sorbonne Universit\'e, \textsc{cnrs}, \textsc{lip6}, F-75005 Paris, France.
    \email{Antoine.Genitrini@lip6.fr}}

\begin{document}

\maketitle

\begin{abstract}
\small\baselineskip=9pt 
Any Boolean function corresponds with a complete full binary decision tree.
This tree can in turn be represented in a maximally compact form as a
direct acyclic graph where common subtrees are factored and shared,
keeping only one copy of each unique subtree. This yields the celebrated and
widely used structure called reduced ordered binary decision diagram (\textsc{robdd}).
We propose to revisit the classical compaction process to give a new way
of enumerating \textsc{robdd}s of a given size without considering fully expanded trees and the compaction step. Our method also provides an unranking procedure for the set of \textsc{robdd}s. As a by-product we get a random uniform and exhaustive sampler for \textsc{robdd}s for a given number of variables and size. 
\end{abstract}
%

	
\section{Introduction}
	\label{sec:intro}

\begin{wrapfigure}[15]{r}{6cm}
    \vspace*{-1.2cm}
    \begin{center}
        \includegraphics[scale=0.9]{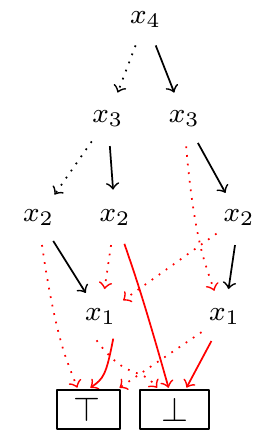}
        \hspace*{0.cm}
        \includegraphics[scale=0.85]{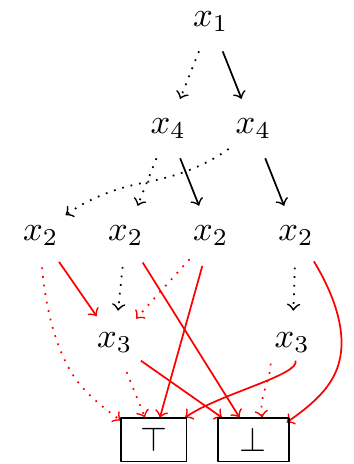}  
        \caption{\label{fig:BDDexample}
            Two Reduced Ordered Binary Decision Diagrams
            associated to the same Boolean function.
            Nodes are labeled with Boolean variables; left dotted edges 
            (resp. right solid edges) are $\False$ links (resp. $\True$ links).}
    \end{center}
\end{wrapfigure}
The representation of a Boolean function as a binary decision tree has been used
for decades. Its main benefit, compared to other representations like a truth table or 
a Boolean circuit, comes from the underlying \emph{divide-and-conquer} paradigm.
Thirty years ago a new data structure emerged, based on the compaction of binary decision tree, 
and hereafter denoted as Binary Decision Diagrams (or \textsc{bdd}s)~\cite{Bryant86}. Its takeoff
has been so spectacular that many variants of compacted structures have been developed, 
and called through many acronyms as presented in~\cite{Wegener00}.
One way to represent the different diagrams consists in their embedding 
as directed acyclic graphs (or \textsc{dag}s).
One reason for the existence of all these variants of diagrams is due to the fact
that each \textsc{dag} correspondence has its own internal agency of the nodes and thus
each representation is oriented towards a specific constraint.
For example, the case of Reduced Ordered Binary Decision Diagrams (\textsc{robdd}s)
is such that the variables do appear at most once and in the same order
along any path from the source to a sink of the \textsc{dag}, and
furthermore, no two occurrences of the same subgraph do appear in the structure. 
For such structures and others, like \textsc{qobdd}s or \textsc{zbdd}s for example,
there is a canonical representation of each Boolean function.

In his book~\cite{Knuth11} Knuth proves or recalls combinatorial results, like
properties for the profile of a \textsc{bdd}, or the way to combine two structures to represent a
more complex function.
However, one notes an unseemly fact. There
are no results about the distribution of the Boolean functions according to their \textsc{robdd} size.
In fact in contrast to (e.g.) binary trees where there is a
recursive characterization that allows to well specify the trees, we have
no local-constraint here for \textsc{robdd}s ans thus a similar recurrence is unexpected.
Very recently, there is a first study exploring experimentally, numerically,
and theoretically the typical and worst-case \textsc{robdd} sizes in~\cite{NV19}.
We aim at obtaining the same kind of combinatorial results but here we design a partition
of the decision diagrams that allows us to go much further in terms of size. 
In particular we obtain an exhaustive enumeration of the diagrams according to their size up to 9 variables.
This was unreachable through the exhaustive approach proposed in~\cite{NV19} due 
to the double exponential complexity of the problem: there are $2^{2^k}$ Boolean functions with $k$ variables.
Our \textsc{c++} implementation 
fully manages the case of $9$ variables~(see~Fig~\ref{fig:distribution}) that corresponds to $2^{512}\approx 10^{154}$ functions.
In particular for 9 Boolean variables, our implementation shows one seventh of all \textsc{robdd}s
are of size 132 (the possible sizes range from $3$ to $143$).
Furthermore, \textsc{robdd}s of size between $125$ and $143$ represents more than $99.8\%$ of all \textsc{robdd}s,
in accordance with theoretical results from~\cite{BV05,GROPL2004}. 

\begin{wrapfigure}[14]{r}{5.5cm}
    \vspace*{-1.1cm}
    \begin{center}
        \includegraphics[scale=0.3]{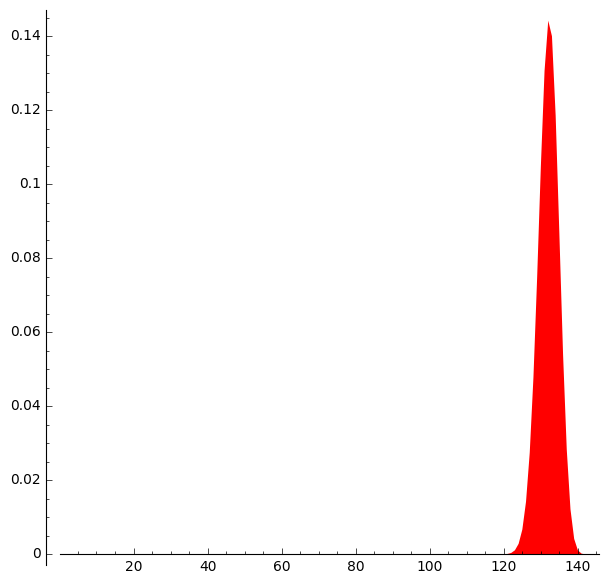} 
    \end{center}
    
    \vspace*{-2mm}
    
    \caption{Proportion of \BDDs\ over 9 variables according to their size
        \label{fig:distribution}}
\end{wrapfigure}
Starting from the well-known compaction process (that takes
a binary decision tree and outputs its compacted form, the \textsc{robdd}),
our combinatorial study gives a way of construction for \textsc{robdd}s of a given size, but 
without the compaction step. We further define a total order over the set of \textsc{robdd}s
and we propose both an unranking and an exhaustive generation algorithm. The first one
gives as a by-product a uniform random sampler for \textsc{robdd}s of a given number of variables
and size.
One strength of our approach is that it allows to sample \emph{uniformly} \textsc{robdd}s of ``small'' size,
for instance of linear size w.r.t the number $k$ of variables,  very efficiently in contrast to a naive rejection algorithm.
The usual uniform distribution on Boolean functions~\cite{BV05} yields with high probability \textsc{robdd}s 
of near maximal size of order $2^k/k$, although \textsc{robdd}s encountered in applications, when tractable, are smaller.
As a perspective, once the unranking method is well understood, and in particular the poset underlying the \textsc{robdd}s,
then we might be able to bias the distribution to sample only in a specific subclass, e.g. \textsc{robdd}s corresponding
to a particular class of formulas (e.g. read-once formulas). 



Our results have practical applications in several contexts, in particular
for testing structures and algorithms. The study given in~\cite{DHT03}
executes tests for an algorithm whose parameter is a binary decision diagram.
It is based on QuickCheck \cite{Claessen:2000:QLT}, the famous software, taking as an entry a random generator
and generating test cases for test suites. Using our uniform generator, we aim at obtaining
statistical testing, in the sense that the underlying distribution of the samples is uniform, thus
allowing to extract statistics thanks to the tests.
Another application of our approach allows to derive exhaustive testing for small structures,
like the study in~\cite{MADKR03}, that we also can conduct inside QuickCheck.

In this paper,  we focus exclusively on \textsc{robdd}s which is 
one of the first and simplest variants. Section~\ref{sec:context} introduces the combinatorics underlying the 
decision tree compaction, leading in Section~\ref{sec:decompo} to a way to unambiguously specify the structure of reduced ordered binary decision diagrams.
We apply this strategy in Section~\ref{sec:algo} and obtain an unranking algorithm for \textsc{robdd}s. 

\section{Decision diagrams as compacted trees}
	\label{sec:context}

This section defines precisely our combinatorial context. Many definitions
are detailed in the monograph of Wegener~\cite{Wegener00} 
and in the dedicated volume~\cite{Knuth11} of Knuth.

\begin{wrapfigure}[14]{r}{6.5cm}
    \vspace*{-1cm}
    \begin{center}
        \includegraphics[scale=0.45]{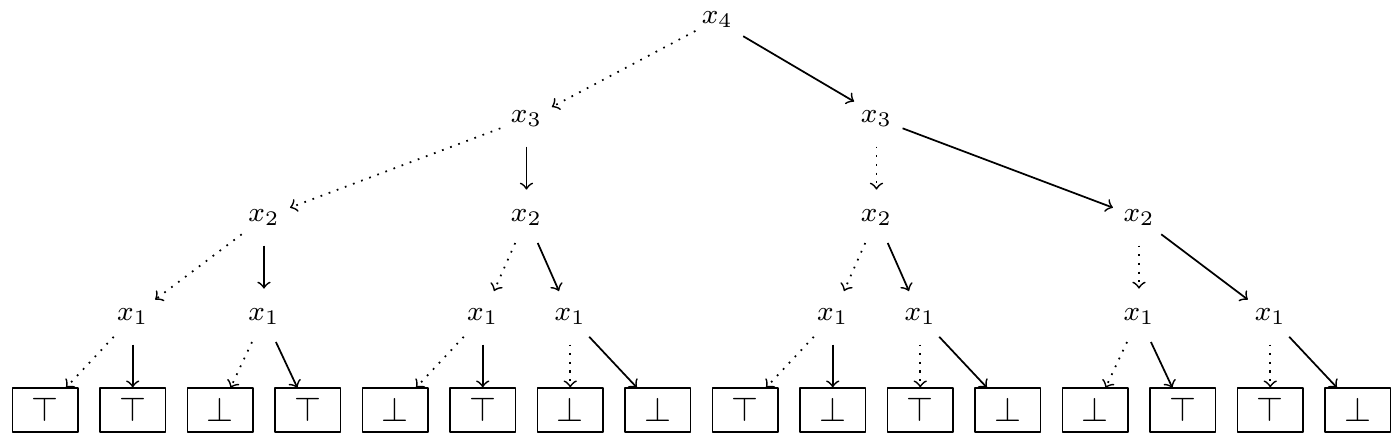}
        
        \vspace*{2mm}
        
        \includegraphics[scale=0.57]{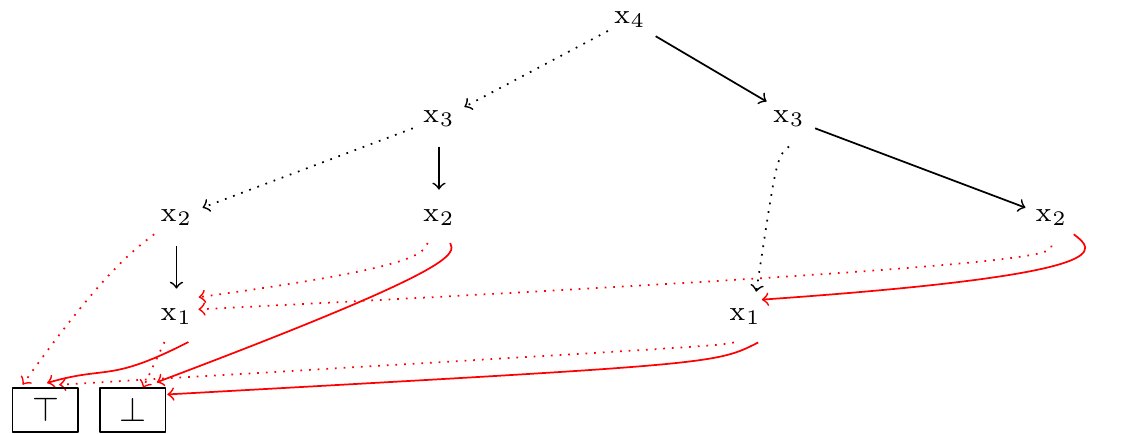}
        
        \vspace*{-2mm}
        
        \caption{\label{fig:example}
            A decision tree and its postorder compaction}
    \end{center}
\end{wrapfigure}
In this section we first recall a one-to-one correspondence between the representation of a Boolean function
as a binary decision tree (built on a specific variable ordering and seen as a \emph{plane tree}, i.e. 
the children of an internal nodes are ordered)
and a reduced ordered binary decision diagram (\textsc{robdd}s) also seen as a plane structure.
This approach is \emph{non-classical} in the context of \textsc{bdd}s, but it allows the formalization of
an equivalence relationship on \textsc{robdd}s that is the key of our enumeration:
in fact our approach foundation relies on breaking down the symmetry in \textsc{robdd}s.
We consider Boolean function on $k$ variables. We recall there are $2^{2^k}$ such Boolean functions.
The compaction process is now formalized.

\par\smallskip
\noindent\textbf{Compaction and plane decision diagrams}
A Boolean function can be represented thanks to a binary decision diagram, which is a rooted, directed, acyclic graph,
which consists of decision nodes and terminal nodes. 
There are two types of terminal nodes $\Tnode$ and $\Fnode$ corresponding
to truth values (resp. $\True$ and $\False$ ).
Each decision node $\nu$ is labeled by a Boolean variable $x_\nu$ and has
two child nodes (called \emph {low child} and \emph{high child}).
The edge from node $\nu$ to a low (or high) child represents
an assignment of $x_\nu$ to $\False$ 
and is represented as a dotted line (respectively $\True$,
represented as a solid line).
In the following we represent \textsc{robdd}s and decision trees (or \textsc{bdd}s in general) as \emph{plane structures},
 i.e. for a node we consider its low child to be its left child and the high child to be its right child.

In a (plane) full binary decision tree, no subtree is shared. By contrast we may decrease 
the number of decision nodes by factoring and sharing common substructures.
Representing a function with its full decision tree is not space efficient.
In Fig.~\ref{fig:example} we depict, on top, a decision tree of a Boolean function on 4 variables.
In the bottom of the figure we represent the compaction of the latter decision tree by using the classical common subexpression recognition notion (cf. e.g.~\cite{FSS90,GGKW20}) based on a postorder traversal of the tree. 

\begin{definition}[Compaction]
Let $T$ be the (plane) binary decision tree of a function $f$.
The \DAG~$T$ is modified through a postorder traversal. When the node $\nu$ is under visit, $\nu$ being a child of a node $\rho$.
If an identical subtree than $T_\nu$, the one rooted in $\nu$, has already be seen during the traversal, rooted in a node $\mu$, then $T_\nu$ is removed from $T$ and the node $\rho$ gets a pointer to $\mu$ (replacing the edge to $\nu$). Once $T$ has been traversed, the resulting \DAG~is the plane \textsc{robdd} of $f$.
\end{definition}
In our figures of \textsc{robdd}s we draw the pointers in red (there is an exception for the edges to the terminal
nodes as we remark in Fig.~\ref{fig:example} also drawn in red).

In a classical setting, \textsc{robdd}s are obtained by applying repetitively reduction rules
(well detailed in~\cite{Wegener00}) to \textsc{obdd}s, and the process is confluent.
Our approach conceptually takes as a starting point a full decision tree with a given ordering on variables
(meaning all nodes at the same level are labeled by the same variable) and applies the compaction rules by examining nodes of the tree in postorder.

For example the plane \textsc{robdd} in Fig.~\ref{fig:example} corresponds to the leftmost \textsc{robdd}
depicted (in the classical way) in Fig.\ref{fig:BDDexample}.
Note that for a given Boolean function, using two distinct variable orderings
can lead to two \textsc{robdd}s of different sizes (see Fig.~\ref{fig:BDDexample} for such a situation).
Nonetheless, an ordering of the variables being fixed, each Boolean function is represented
by exactly one single \textsc{robdd} obtained through the compaction of its decision tree for this order.

\par\medskip
\noindent \emph{In the rest of the paper, we consider only plane \textsc{robdd}s.
From now we thus call them \textsc{bdd}s. We also assume the set of variables $X=\{\ldots > x_k> \ldots> x_1\}$ is totally ordered.}
\medskip

Our first goal aims at giving an effective method to enumerate \BDDs\ with
a chosen number $k$ of variables and size $n$.
A first naive approach is:
(1) enumerate all the  $2^{2^k}$ Boolean functions by construction of the decision trees;
(2) apply the compaction procedure;
and (3) finally filter the \BDDs\ of size equal to the target size $n$.
This algorithm ceases to be practical for $k$ larger than $4$ (see~\cite{NV19}).

In this paper, we propose a new combinatorial description of \BDDs\ providing 
the basis for an enumeration algorithm avoiding the enumeration 
of all Boolean functions on $k$ variables.

\section{Recursive decomposition}
	\label{sec:decompo}

This section introduces a canonical and unambiguous decomposition
of the \BDDs\ yielding a recursive algorithm for their enumeration.

\par\smallskip
\noindent\textbf{Automaton point of view}
Let us introduce an equivalent representation for a \BDD.
A \BDD~can indeed be described as a deterministic finite automaton with additional
constraints and properties. This point of view gives a convenient formal characterization of
the decomposition of \BDDs\ used in our algorithms.
\begin{definition}[BDD as an automaton]
\label{def:automaton}
    A \BDD\ $B$ of index $k$ is a tuple $(Q, I, r, \delta)$ where
    \begin{itemize}[topsep=0pt]
        \item
        $Q$ is the set of nodes of the \BDD. $Q$ contains two special sink nodes $\Fnode$ and~$\Tnode$.
        \item
        $I: Q \to \{0, \dots, k\}$ is the index function which associates with every node its index.
        By convention the index of both sink nodes is $0$.
        \item
        $r\in Q$ is the root and has index $I(r)=k$.
        \item
        $\delta: Q\setminus\{\Fnode, \Tnode\} \times \{\False, \True\} \to Q$ is the full transition function.
    \end{itemize}
    There are constraints on $\delta$ translating the classical ones of the \BDDs:
    \begin{itemize}[topsep=0pt]
        \item for any node $\nu \in Q\setminus\{\Fnode, \Tnode\}$,
            $\delta(\nu, \False) \neq \delta(\nu, \True)$.
        \item for any distinct nodes $\mu$ and $\nu$ with the same index, we have 
            $\delta(\mu, \False) \neq \delta(\nu, \False)$ or $\delta(\mu, \True) \neq \delta(\nu, \True)$.
        \item the graph underlying $\delta$ forms a \DAG\ with a unique node of in-degree 0, the root $r$.
        \item if $\tau = \delta(\nu, \alpha)$ for some $\alpha \in \{\False, \True\}$
            then $I(\tau) < I(\nu)$.
    \end{itemize}
    We say $\tau$ is the low child of $\nu$ (respectively high child of $\nu$)
    if $\delta(\nu, \False) = \tau$ (resp. $\delta(\nu, \True)=\tau$).
\end{definition}

\begin{definition}[Spine of a BDD, tree and non-tree edges]
Let a \BDD\ $B=(Q, I, r, \delta)$ of root-index $k$.
The \emph{spine} of $B$ is the spanning tree obtained by a \emph{depth-first search}
of the (plane) \BDD\ (where low child is accessed before the high one), and omitting the sinks $\Fnode$ and $\Tnode$.
For a \BDD\ $B$, the edges of the spine forms the set
of \emph{tree edges} (drawn in black). The other edges
form the set of \emph{non-tree edges} (drawn in red).
We describe the spine $T$ as a tuple $T=(Q', I, r, \delta')$ with
set of nodes $Q'= Q \setminus\{\Fnode, \Tnode\}$
(with the same index function $I$ as for $B$). 
The edges of the spine are described using a partial transition
function $\delta': Q' \times \{\False, \True\} \to Q'\cup\{\nil\}$
where $\nil$ is a special symbol designating an undefined transition.
\end{definition}
\begin{figure}[h!]
    \centering{\includegraphics[width=.45\textwidth]{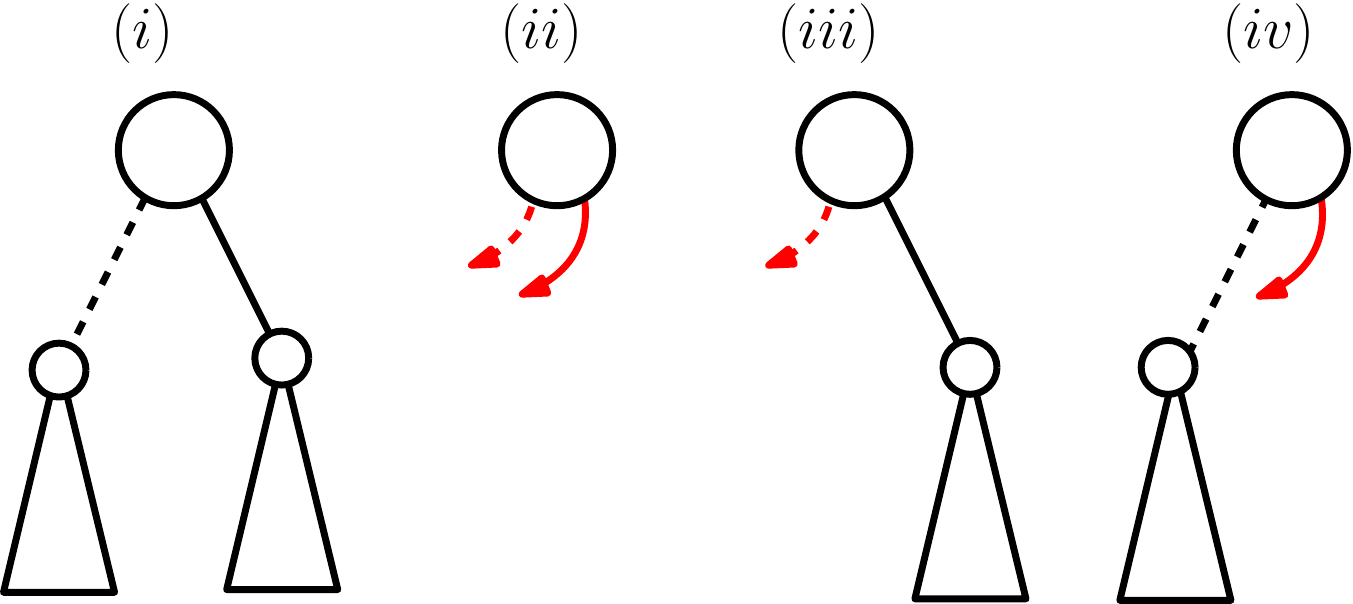}}    
\caption{\label{fig:types} The four cases for a node of a spine. From left to right: an internal node with both transitions defined, two half edges, one low (left) half edge, one high (right) half edge (cf. Proposition~\ref{prop:no-tree})} 
\end{figure} 
Using standard terminology for depth-first search, \emph{non-tree} edges are either \emph{forward or cross} edges. We remark that, by definition, a \textsc{dag} admits no cycles and still in the standard notation, it has no \emph{backward} edges.

Undefined values of the transition function $\delta'$ can conveniently
be seen as half edges. Since in a \BDD\ every non-sink node has two children,
the spine of a \BDD\ of size $n$ has $(n-2)$ nodes and $(n-1)$ half edges
(drawn in red in Fig.~\ref{fig:BDDexample}). The four possible types of a node are depicted, as the roots in Fig.~\ref{fig:types}.

\begin{definition}[Valid tree]
    A binary tree is said to be \emph{valid} if it is the spine of some \BDD.
    The set of spines of size $n$ is denoted as $\mathcal{T}_n$.
\end{definition}
See Fig.~\ref{fig:spines} for examples of valid and invalid trees.
To the best of our knowledge, there is no way to characterize valid trees, apart from exhibiting a \textsc{robdd} admitting this tree as a spine. We will discuss this point later.

For enumerating \BDDs\ it will prove convenient to introduce the profile list of a set of nodes
and some other useful notation for lists manipulation.
\begin{definition}
    The \emph{profile} of $\mathcal{N}$, denoted by $\profile(\mathcal{N})$,
    is a list with $(k+1)$ components $\boldsymbol{p}=(p_0, \dots, p_k)$
    where $k=\max_{\nu\in \mathcal{N}} \idx(\nu)$ is the maximal index and $p_i$
    is the number of nodes of index $i$ in $\mathcal{N}$.
\end{definition}
This definition extends naturally to trees, graphs, etc. We also equip the set of lists with a `$+$' operation:
let two lists $\boldsymbol{v}=(v_0, \dots, v_m)$ and
$\boldsymbol{v}'=(v'_0, \dots, v'_{n})$ with $n\ge m$ (w.l.o.g.), 
the sum $\boldsymbol{v}+\boldsymbol{v}'$ is equal to
$\boldsymbol{w}=(w_0, \dots, w_n)$ where
for all $0 \le i \le m,\;  w_i = v_i+v'_i$ and otherwise, when $m< i \le n,\; w_i=v'_i$.


\medskip
In the following, we will use two orderings on the nodes of a plane \BDD\ induced by depth-first search,
and called postordering and preordering.
Since the structure is plane these orderings correspond exactly with the classical postorder traversal
and the preorder traversal of its spanning tree. 
In a tree, for a node $\nu$ with 
low child $\nu_\False$ and high child $\nu_\True$, the postorder traversal visits the subtree
rooted at $\nu_\False$ then, the one rooted at $\nu_\True$ and finally $\nu$. 
The preorder traversal
first visits the node $\nu$, then the subtree rooted at $\nu_\False$
and finally the subtree rooted at $\nu_\True$.
We use the notation $\mu \postorder \nu$
(resp. $\mu \preorder \nu$) 
if the node $\mu$ is visited before $\nu$ using the postorder (resp. preorder) traversal.

We characterize now how the partial transition function of the spine is 
related to the full transition function of the \BDD. 
Introducing the pool and level set of a node, we describe the valid choices
for non-tree edges to yield a \BDD.
\begin{definition}[Pool and level set]
    Let $T$ be the spine of a \BDD. The pool of a node $\nu \in T$ is
    \[
	    \mathcal{P}_T(\nu) = \left\{ \tau  \in T \mid \tau \preorder \nu \text{ and } \idx(\tau) < \idx(\nu)\right \} \cup \left\{\Fnode, \Tnode \right\}.
    \]
    The pool profile $p_T(\nu)$ of a node $\nu$ in a spine $T$ is
    $p_T(\nu) = \profile(\mathcal{P}_T(\nu))$.\\
    The \emph{level set} of $\nu$ is $\mathcal{S}_T(\nu) = \{ \tau \in T \mid \tau \preorder \nu \text{ and } \idx(\tau) = \idx(\nu)\}$,
    and the \emph{level rank} $s_T(\nu)=\len{\mathcal{S}_T(\nu)}$ of a node $\nu$ is the rank of $\nu$
    among the set of nodes with the same index.
\end{definition}
Informally the pool of a node $\nu$ of a tree $T$ is the set of nodes we could choose
as a low child for $\nu$ without invalidating the spine.
The first component of a pool profile is always~$2$ since both sinks $\Fnode$ or $\Tnode$
are present in the pool of any node of the spine 
(providing the underlying \BDD\ is not reduced to $\True$ or $\False$).

\begin{proposition}
\label{prop:no-tree}
	Let $T=(Q', I, r, \delta')$ be a valid spine with set of nodes $Q'$, root~$r$
	and partial transition function $\delta': Q' \times \{\False, \True\} \to Q' \cup \{\nil\}$.
	The full transition function 
	$\delta: Q' \times \{\False, \True\} \to Q' \cup \{\Fnode, \Tnode\}$ is the
	transition function of a \BDD\ with spine $T$ if and only for any node $\nu\in Q'$,
	noting $\nu_\False=\delta(\nu, \False)$ and $\nu_\True=\delta(\nu, \True)$,
	the pair $(\nu_\False, \nu_\True)$ satisfies
	\begin{itemize}[topsep=0pt]
	    \item[(i)] if $\delta'(\nu, \False)\neq\nil$ and $\delta'(\nu, \True)\neq\nil$
	        then $\nu_\alpha = \delta'(\nu, \alpha)$ for $\alpha \in\{\False, \True\}$.
	    \item[(ii)] if $\delta'(\nu, \False)=\delta'(\nu, \True)=\nil$, then
	        \[
	            \nu_\alpha \preorder \nu \text{ and } \idx(\nu_\alpha) < \idx(\nu) \text{ for } \alpha \in\{\False, \True\} \text{ and }
	            \nu_\False \neq \nu_\True,
	        \]
	        and there is no node $\tau\neq \nu$ with the same index as $\nu$ such that
	        $\delta(\tau, \cdot) = \delta(\nu, \cdot)$.
	    \item[(iii)] if $\delta'(\nu, \False)=\nil$ and $\delta'(\nu, \True)\neq\nil$, then
	        \[
	            \nu_\False \preorder \nu, \ \idx(\nu_\False) < \idx(\nu) \text{ and } \nu_\True = \delta'(\nu, \True).
	        \]
	    \item[(iv)] if $\delta'(\nu, \False)\neq\nil$ and $\delta'(\nu, \True)=\nil$, then
	        $\nu_\False = \delta'(\nu, \False)$ and 
	        \[
	            \nu_\True \postorder \nu, \ \nu_\True \neq \nu_\False \text{ and } \idx(\nu_\True) < \idx(\nu).
	        \]
	\end{itemize}
\end{proposition}
\begin{proof}
Since $\delta(\cdot, \cdot)$ must extend  $\delta'(\cdot, \cdot)$, case $(i)$ is trivial
since we must only extend the transition function where $\delta(\nu, \alpha)=\nil$.
In case $(ii)$, we have to choose for $(\nu_\False, \nu_\True)$ two nodes in the pool of $\nu$
($\nu$ is an external node of the spine). We use the preorder traversal
(but since $\nu$ is an external node, the postorder would also be fine).
Moreover $\nu_\False \neq \nu_\True$ and no node with the same index as $\nu$ can have the exact same descendants $(\nu_\False, \nu_\True)$ in accordance with Definition \ref{def:automaton}.
In case $(iii)$, the low child must be chosen in the pool of $\nu$ since we preserve the spine.
In case $(iv)$, the high child of $\nu$ is also chosen in the pool $\nu$ or in the descendants of $\nu_\False$
in the spine $T$ (still different from $\nu_\False$ by Definition \ref{def:automaton}). \Qed
\end{proof}

\section{Counting and Generating BDDs}
	\label{sec:algo}

In this section, we sketch algorithms in order to count and sample \textsc{bdd}s
of a given size $n$ and given number $k$ of variables.
\subsubsection*{Counting BDDs}
Given a spine $T$, we can compute the number
of \BDDs\ corresponding with this spine.
Thus counting \BDDs\ of a certain size $n$ will consists in building
all valid spines of size $(n-2)$ and completing the transition function
of the spine in all possible ways according to Proposition \ref{prop:no-tree}.

\begin{definition}[Weight]
    Let $T=(Q', I, \delta', r)$ be a spine, the \emph{weight} $w_T(\nu)$ of a node $\nu \in Q'$ is the number
    of possibilities for completing the transition function $\delta'(\mu, \cdot)$ and yielding 
    a \BDD\ with spine $T$. 
    The \emph{cumulated weight} of a subtree $T_\nu$ rooted at $\nu\in T$ is
    $W_T(\nu) = \prod_{\tau\in  T_\nu} w_T(\tau)$. We write $W(T)=W_T(r)$ 
    to denote the cumulated weight of the whole spine $T$ rooted at $r$.
\end{definition}
Note that the number of choices for the missing transitions out of a node $\nu$ are the ones
remaining after previous choices have been made for other nodes of the spine.
\begin{proposition}[Weight of a node]
\label{prop:weight_of_a_node}
    Let $T$ be a spine $T$, the \emph{weight} of a node $\nu \in T$ is 
    \[
	    w_T(\nu) =
		    \begin{cases}
		    1 & \text{if $\delta'(\nu, 0)\neq \nil$ and $\delta'(\nu, 1)\neq \nil$}\\
		    \norm{p_T(\nu)}\big(\norm{p_T(\nu)}-1\big) -s_T(\nu) & \text{if $\delta'(\nu, 0)=\delta'(\nu, 1)= \nil$}\\
		    \norm{p_T(\nu)+\profile(T')} & \text{if $\delta'(\nu, 0)\neq \nil$ and $\delta'(\nu, 1)= \nil$}\\
		    \norm{p_T(\nu)}& \text{if $\delta'(\nu, 0)= \nil$ and $\delta'(\nu, 1)\neq \nil$}\\
		    \end{cases}
    \]
    where $p_T(\nu)$ is the pool profile of node $\nu$, $T'=T_{\nu_0}$ is the subtree (when defined) rooted at $\nu_0 = \delta'(\nu, \False)$, and, for a list $\boldsymbol{p}=(p_0, \dots, p_k)$, we denote $\norm{\boldsymbol{p}}= \sum_{i=0}^k p_i$.
\end{proposition}
In the third case, by $p_T(\nu)+\profile(T')$, we mean the profile of the set of nodes visited before $\nu$ with the postorder traversal of $T$ and of index strictly smaller than $\idx(\nu)$.
\begin{proof}
This is a direct application of Proposition \ref{prop:no-tree}. \Qed
\end{proof}

This formula allows to detect if a tree is a valid spine. 
Indeed as soon as the weight of a node is zero or negative, 
there is no way to define a total transition function $\delta$ for a $\BDD$. 
Note that this situation can only happen for external nodes
having two half edges, since for any node $\nu\in Q'$ 
and any spine $T$, $\norm{p_T(\nu)} \ge 2$.
\begin{wrapfigure}[9]{r}{0.55\textwidth}
    \vspace*{-4mm}
    \centering{\includegraphics[width=0.54\textwidth]{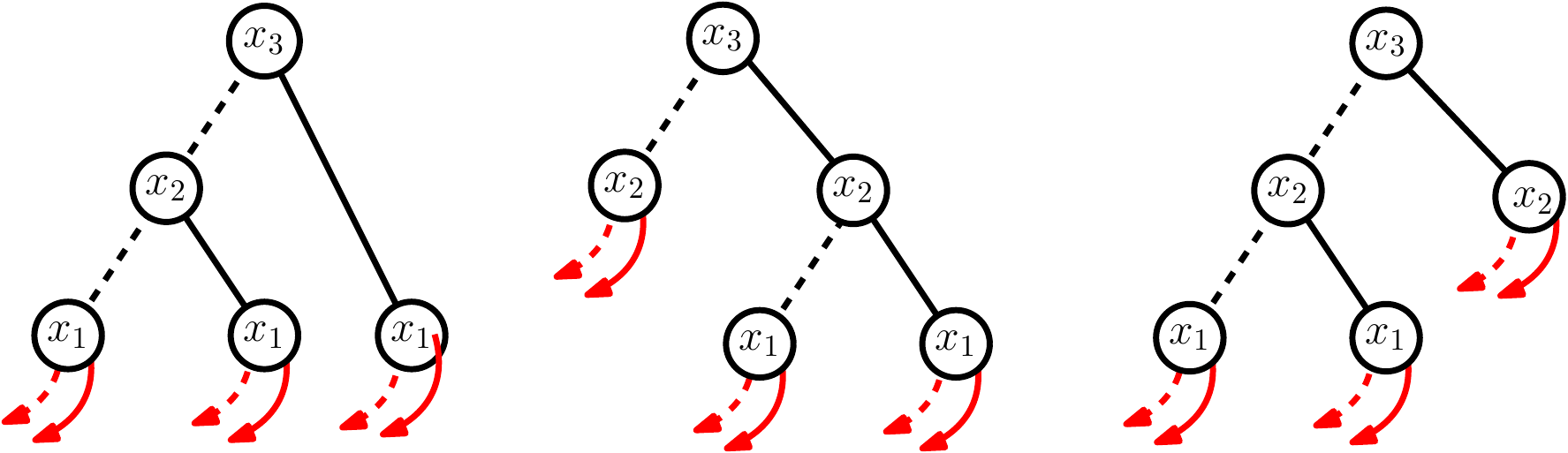}}
    
   \vspace*{-2mm}
    \caption{\label{fig:spines} three examples of binary trees (first one is invalid, the two other have respective weights $4$ and~$24$).}
\end{wrapfigure}
In Fig. \ref{fig:spines}, the binary tree on the left is invalid and cannot be the spine of any \BDD. The two other trees on the right have weights $4$ and $24$, i.e., are resp. the spines of exactly $4$ and $24$ \BDDs. It is an open problem to characterize the set of valid trees (apart from exhibiting corresponding \BDDs).
 
Proposition \ref{prop:weight_of_a_node} gives access to the total weight
of the spine $W(T)$ using a recursive procedure. 
A natural way to proceed
algorithmically is to use a recursive postorder traversal 
of the tree maintaining at each node the weight in a multiplicative manner.
To do so we need to keep track in the traversal of the pool profile and
level rank of the current node.

Initially the pool of the root is reduced to the set $\{\Fnode, \Tnode\}$.
Thus the initial pool profile of the root of index $k$
is initialized to $(2, 0, \dots, 0)$ of length $k$. 
The level rank of the root of the spine is $0$.

\begin{proposition}
\label{prop:count}
    Let $N(n, k)$ be the number of \BDDs\ of index $k$ and size $n$
    \[ \textstyle
    	N(n, k) = \sum_{T \in  \mathcal{T}_{n-2, k}} W(T),
    \]
    where $\mathcal{T}_{m, k}$ is the set of valid spines with $m$ nodes 
    for \BDDs\ of index $k$. 
\end{proposition}
\begin{proof}
The weight of a spine is the number of ways of extending
the transition function of $T$ (Proposition~\ref{prop:no-tree}),
hence the number of \BDDs\ for this given spine. \Qed
\end{proof}

\subsubsection*{Combinatorial description of spines}
The set of spines is not straightforward to characterize
in a combinatorial way. Indeed we need context to decide if
the weight of a particular node in a tree is $0$ or less, which in turn yields that the tree is not valid .
To enumerate spines, we build recursively binary trees, and, while computing weights for its nodes, as
soon we can decide the (partially built) tree is not valid, the tree is discarded.

To decompose (or count) spines  of any size or index,
$\mathcal{T} = \bigcup_{n\ge 1} \bigcup_{k\ge 1} \mathcal{T}_{n, k}$,
we introduce a partition over subtrees which can occur in a spine
$T \in \mathcal{T}$. The goal is to
identify identical subtrees occurring within different spines and with the 
same weight to avoid redundant computations.

The combinatorial description we are about to present originates
from the following observation: let us fix a spine $T$ and a node $\nu \in T$.
From Proposition~\ref{prop:weight_of_a_node}, to compute the 
cumulated weight of the subtree $T_\nu$ rooted at $\nu$, 
the sole knowledge of the pool profile $p_T(\nu)$ and the level rank 
$s_T(\nu)$ is sufficient.

Let $S$ and $S'$ be two subtrees with respective roots $\nu$ and $\nu'$
in some spines $T$ and $T'$, we denote $S \equiv S'$ if 
the following three conditions are satisfied:
\begin{itemize}[topsep=0pt]
\item both trees have the same size: $\len{S} = \len{S'}$;
\item the roots of both trees have the same pool profile: $p_T(\nu) = p_{T'}(\nu')$;
\item the roots of both trees have the same level rank: $s_T(\nu)=s_{T'}(\nu')$.
\end{itemize}
The set $\mathcal{T}_{m, \boldsymbol{p}, s}$ is the class equivalence
for the relation `$\equiv$' and gathers trees (as a set, without multiplicities)
which are possible subtrees of size $m$ in any spine, 
knowing only the pool profile $\boldsymbol{p}$ and level rank $s$ 
of the root of the subtree. More formally:
\[
\mathcal{T}_{m, \boldsymbol{p}, s} = \{ T_\nu \mid (\exists T \in \mathcal{T}) 
        \; (\exists \nu \in T) \ p_T(\nu) = \boldsymbol{p} \text{ and } s_T(\nu)=s \}.
        \]
Note that we have $ \mathcal{T}_{n, k} = \mathcal{T}_{n, (2, 0,\dots, 0), 0},$ where $(2, 0, \dots, 0)$ has $k$ components.

\begin{proposition}
\label{prop:decompo}
    The set $\mathcal{T}_{m, \boldsymbol{p}, s}$
    of subtrees of size $m$ rooted at a node having pool profile $\boldsymbol{p}=(p_0, \dots, p_{k-1})$ 
    and level rank $s$ occurring in the set of spines $\mathcal{T}$
    is decomposed without any ambiguity.
    We decompose a subtree $T\in \mathcal{T}_{m, \boldsymbol{p}, s}$ as a tuple $(\nu, T', T'')$ where the root $\nu$ has index $k$ and $T'$  and $T''$ are its left and right (possibly empty) subtrees of respective sizes $i$ and $m-1-i$, with $0\le i\le m-1$, and verifying (when non empty)
    \begin{align*}
    \textit{(i)}\ & T' \in \bigcup_{k_0 \in \{1, \dots,k-1\}} \mathcal{T}_{i, (p_0, \dots, p_{k_0 -1}), p_{k_0}}\\
    \textit{(ii)}\ & T'' \in \bigcup_{k_1 \in \{1, \dots,k-1\}}\mathcal{T}_{m-i-1, (p'_0, \dots, p'_{k_1 -1}), p'_{k_1}}, \text{ with $\boldsymbol{p}'=\boldsymbol{p}+\profile(T')$}\\
    \textit{(iii)}\ & \text{if $m=1$ then  $\textstyle\big(\sum_{i=0}^k p_i\big) \cdot \big(-1+\sum_{i=0}^k p_i\big) -s  >0$}.
    \end{align*}
\end{proposition}
This proposition ensures that we can decompose unambiguously subtrees
occurring in spines in accordance with the equivalence relation `$\equiv$'.
Practically this means that instead of considering all possible subtrees 
for all possible spines, we can compute cumulative weights for each 
representative of the equivalence relation 
(which are fewer although still of exponential cardinality).

Algorithm \textnormal{\funcs{count}{$n, \boldsymbol{p}, s$}} in Algorithm~\ref{algo:count} 
enumerates spines of \BDDs\ and,  at the same time, computes their cumulated weights. 
It takes as arguments a size~$n$ for considering all subtrees of size $n$,
assuming an initial pool profile $\boldsymbol{p}=(p_0, \dots, p_{k-1})$, 
level rank $s$ and index $k$ for the root of these trees.
It returns in an associative array 
a list of pairs $(\boldsymbol{t}, w)$ where 
\begin{itemize}[topsep=0pt]
\item
$\boldsymbol{t}=(t_0, t_1, \dots, t_{k-1}, t_k)$ ranges over the set of profiles of trees in $\mathcal{T}_{m, \boldsymbol{p}, s}$, i.e., $\boldsymbol{t} \in\{ \profile(T) \mid T \in \mathcal{T}_{m, \boldsymbol{p}=(p_0, \dots, p_{k-1}), s}\}$.
\item
$w$ is the sum over all equivalent trees of size $m$ with profile 
$\boldsymbol{t}$ of their cumulated weights when the root has pool profile
$\boldsymbol{p}$ and level rank $s$ (which gives enough information
to compute the cumulated weight for each tree 
using Proposition~\ref{prop:weight_of_a_node}).
\end{itemize}
Note that any subtree $T$ with a root of index $i$ has a profile
$\boldsymbol{t}=(t_0, \dots, t_i)$ with $t_0=0$ and $t_i=1$.
\begin{proposition}
The number $N(n, k)$ of \BDDs\ of size $n$ and of index $k$ is computed thanks to Algorithm \textnormal{\funcs{count}{}} and is equal to
\[
N(n, k) = \sum_{(\boldsymbol{t}, w)\in \textnormal{\funcs{count}{$n-2, (2, 0, \dots, 0), 0$}}} w,
\]
where $(2, 0, \dots, 0)$ has $k$ components and corresponds with a pool reduced 
to the two sink nodes $\Fnode$ and $\Tnode$ of index $0$.
\end{proposition}
\begin{proof}
Indeed $\boldsymbol{t}$ ranges over all possible profiles for spines of size $(n-2)$ 
and we sum the weights of all spines for these profiles.
Hence we compute exactly the number of \BDDs\ of size $n$. \Qed
\end{proof}

An important refinement for this algorithm is to remark when summing over all spines, 
we consider subtrees of the same size whose root shares the same pool profile
and same level rank, hence the same context. In order to avoid performing the same exact
computations twice (or more) we can use \emph{memoization} technique (that is storing intermediary results). 
It is an important trick to reduce the time complexity, although at the cost of some
memory consumption.

\subsubsection*{Complexity of the counting algorithm}
First, we remark the numbers involved  in the computations are (very) big numbers (as seen before, of order $2^{2^k}$).
\begin{proposition}
\label{prop:complecity_count}
The complexity (in the number of arithmetic operations)
of the computations of the Algorithm~\ref{algo:count} to evaluate $N(n, k)$ is $O\left(\frac{1}{k} 2^{3k^2 / 2 + k} \right)$.
\end{proposition}
For Boolean functions in $k$ variables, although the time complexity of our algorithm is of exponential growth
$2^{3k^2 / 2}$. However the state space of Boolean functions is $2^{2^k}$ thus our computation is still much better than the exhaustive construction.


\begin{algorithm}[htb]
\begin{minipage}{\textwidth}
\begin{algorithmic}
\Function{count}{$n, \boldsymbol{p}=(p_0, \dots, p_{k-1}), s$}
\State $d\gets \{ \,\}$ \Comment{Empty dictionary}
  \For{$i \gets 0$ to $n-1$} \Comment{Left/right subtrees of size $i$/$n-i-1$}
    \State $d_0\gets \{ \, \}$
    \If{$i=0$} 
      $d_0\gets \left\{\boldsymbol{\epsilon} : \sum_{i=0}^{k-1} p_i \right\}$ \Comment{left subtree is empty, see $^*$}
    \Else
      \For{$k_0\gets 1$ to $k-1$} \Comment{left node has index $k_0$}
        \State $d_0 \gets d_0 \cup \text{\funcs{count}{$i, (p_0, \dots, p_{k_0-2}), p_{k_0-1}$}}$
      \EndFor
    \EndIf
    \For{$(\boldsymbol{\ell}, w_0) \gets d_0$}
      \State $d_1\gets \{ \, \}$
      \State $\boldsymbol{p'} \gets \boldsymbol{p}+\boldsymbol{\ell}$
      \If{$n-1-i=0$}
        $d_1 \gets d_1 \cup \left\{ \boldsymbol{\epsilon} : -1+\sum_{i=0}^{k-1} p'_i \right\}$  \Comment{right subtree is empty}
      \Else\
        \For{$k_1\gets 1$ to $k-1$} \Comment{right node has index $k_1$}
          \State $d_1\gets d_1 \cup \text{\funcs{count}{$n-1-i, (p'_0, \dots, p'_{k_1-2}), p'_{k_1-1}$}}$
        \EndFor
      \EndIf
      \For{$(\boldsymbol{r}, w_1) \gets d_1$}
        \State $w \gets w_0 \cdot w_1$
        \If{$n=1$} $w \gets w - s$
        \EndIf
        \If{$w>0$}
          \State $\boldsymbol{t} \gets \boldsymbol{\ell}+\boldsymbol{r}+\boldsymbol{e}^{(k)}$ \Comment{index profile of the subtree}
          \If{$\boldsymbol{t} \in d$} 
            $d[\boldsymbol{t}] \gets d[\boldsymbol{t}]+w$ \Comment{update if $\boldsymbol{t}$ is already a key in $d$}
          \Else\ 
            $d\gets d \cup \left\{\boldsymbol{t}: w \right\}$ \Comment{$\boldsymbol{t}$ is a new key in $d$}
          \EndIf
        \EndIf
      \EndFor
    \EndFor
  \EndFor

\State \Return $d$
\EndFunction
\end{algorithmic}
\smallskip
\scriptsize{$^*$
For an integer $k\ge 0$, the list $\boldsymbol{e}^{(k)}=(0, \dots, 0, 1)$
is the list with $(k+1)$ components where the last entry is $1$ and all
others are $0$. The empty list of size $0$ is denoted~$\boldsymbol{\epsilon}$.
}
\caption{\label{algo:count} Algorithm \funcs{count}{}. The initial pool profile of the root (of index $k$) is $(2, 0, \dots, 0)$ of length $k$.}
\end{minipage}
\end{algorithm}
%


\subsubsection*{Unranking BDDs}

Using the classical recursive method for the generation of structures~\cite{WN89}
we base our generation approach on the combinatorial counting approach.
Since the class of objects under study 
seems not admissible in the sense given in Analytic Combinatorics~\cite{FlajoletSedgewick2009}, we cannot directly
apply the advanced techniques presented in~\cite{FZC94} nor the approaches
by Mart{\'i}nez and Molinero~\cite{MM01}.
Thus we devise an unranking algorithm for \BDDs\ 
and get as by-products algorithms for uniform random sampling
and exhaustive generation.

The ranking/unranking techniques for objects of a combinatorial class $\mathcal{C}$ 
of size $N$ consists in building a bijection between any $c\in \mathcal{C}$ and an integer
(its \emph{rank}) in the interval $[0 \pp N-1]$ (if we starts from $0$).
This leads trivially to a uniform sampling algorithm by drawing uniformly first an integer and then building the corresponding object.
%

\begin{proposition}
    Once the pre-computations are done, the unranking (or uniform random sampling) algorithm
    needs $O\left(n \cdot | \mathcal{T}_{n, k} | \right)$ arithmetic operations to build a
    \BDD\ of index $k$ and size $n$.
\end{proposition}
First, remark that the worst case happens when $n$ is of order the largest
possible size of a \BDD\ over $k$ variables $O(\frac{2^{k/2}}{k})$ (cf.~\cite[p. 102]{Knuth11}) which corresponds to the generic case according to Fig.~\ref{fig:distribution}.
Furthermore, the number of profiles is of order $2^{\frac{k^2}{2}}$.
To generate a \BDD\ given its rank, we first identify the correct profile of its spine (by enumeration).
Then according to this target profile, recursively, for each node, we traverse at most all spines with this profile,
in order to decompose the substructures in its left and right part, yielding the upper bound.

\par\smallskip

\emph{As a conclusion,
note the process of enumerating, counting and sampling we introduced can be adapted 
to subclasses of functions (for instance those for which all variables are essential), 
but also to other strategies of compaction, like those used for Quasi-Reduced \BDDs\ and  Zero-suppressed \BDDs.
A natural question is also to provide an algorithm enumerating valid spines and not all invalid ones as well to get more efficient enumeration and unranking algorithms for \textsc{robdd}s.
These questions will be addressed in future work.}

\par\medskip
\noindent\textbf{Acknowledgement.}
We thank the anonymous reviewers whose comments and suggestions helped improve and clarify this manuscript.


%

\bibliographystyle{splncs04}
\bibliography{sections/biblio}

\end{document}